\documentclass[final,3p,times]{elsarticle}

\usepackage{amssymb,amsmath,amsthm,bm,mathtools}
\usepackage{booktabs,threeparttable,multirow,enumitem,xcolor}
\usepackage[colorlinks=true, linkcolor=blue]{hyperref}

\newtheorem{theorem}{Theorem}[section]
\theoremstyle{remark}
\newtheorem{remark}{Remark}[section]

\newcommand\tbbint{{-\mkern -16mu\int}}
\newcommand\dbbint{{-\mkern -19mu\int}}
\newcommand\bbint{\mathop{\mathchoice{\dbbint}{\tbbint}{\tbbint}{\tbbint}\mkern -6mu}\nolimits}

\makeatletter
\newcounter{subfig}[figure]

\newcommand{\makesubfig}[1]{\refstepcounter{subfig}\label{#1}}
\newcommand{\figsubref}[2]{\hyperref[#2]{\ref*{#1}\ref*{#2}}}
\makeatother

\journal{}
\affiliation[label1]{organization={School of Mathematical Sciences, Peking University},
            city={Beijing},
            postcode={100871},
            country={China}}
\affiliation[label2]{organization={School of Mathematical Sciences, Laboratory of Mathematics and Complex Systems, Ministry of Education, Beijing Normal University},
            city={Beijing},
            postcode={100875}, 
            country={China}}
\affiliation[label3]{organization={School of Mathematical Sciences, Beijing International Center for Mathematical Research, Center for Quantitative Biology, Center for Machine Learning Research, Peking University},
            city={Beijing},
            postcode={100871},
            country={China}}
\fntext[fn1]{This work was supported by the National Natural Science Foundation of China (No. 12225102, T2321001, 12288101) and National Key Research and Development Program of China (No. 2024YFA0919500). }
\cortext[cor1]{Corresponding author}

\begin{document}

\begin{frontmatter}

\title{A Geometry-Adaptive Deep Variational Framework for Phase Discovery in the Landau--Brazovskii Model}

\author[label1]{Yuchen Xie}
\author[label2]{Jianyuan Yin\corref{cor1}}
\ead{jyyin@bnu.edu.cn}
\author[label1,label3]{Lei Zhang\corref{cor1}\fnref{fn1}}
\ead{zhangl@math.pku.edu.cn}

\begin{abstract}
The discovery of ordered structures in pattern-forming systems, such as the Landau--Brazovskii (LB) model, is often limited by the sensitivity of numerical solvers to the prescribed computational domain size. 
Incompatible domains induce artificial stress, frequently trapping the system in high-energy metastable configurations.
To resolve this issue, we propose a Geometry-Adaptive Deep Variational Framework (GeoDVF) that jointly optimizes the infinite-dimensional order parameter, which is parameterized by a neural network, and the finite-dimensional geometric parameters of the computational domain.
By explicitly treating the domain size as trainable variables within the variational formulation, GeoDVF naturally eliminates artificial stress during training.
To escape the attraction basin of the disordered phase under small initializations, we introduce a warmup penalty mechanism, which effectively destabilizes the disordered phase, enabling the spontaneous nucleation of complex three-dimensional ordered phases from random initializations. 
Furthermore, we design a guided initialization protocol to resolve topologically intricate phases associated with narrow basins of attraction.
Extensive numerical experiments show that GeoDVF provides a robust and geometry-consistent variational solver capable of identifying both stable and metastable states without prior knowledge.
\end{abstract}


\end{frontmatter}

\section{Introduction}
\label{sec:intro}
The Landau--Brazovskii (LB) model is a prototypical continuum framework for describing phase transitions driven by short-wavelength instabilities in systems that undergo ordering at a finite wavenumber. 
Originally proposed by Brazovskii to describe the transition from an isotropic liquid to a nonuniform state \cite{brazovskii1975phase}, this model has become a standard theoretical framework for investigating ordered patterns in diverse physical systems. 
Its applications span multiple systems, ranging from microphase separation in block copolymers \cite{shi1996theory,zhang2008efficient} and liquid-crystal ordering \cite{kats1993weak,wang2021modelling} to biological pattern formation such as viral capsids, and nuclear pasta structures in neutron stars \cite{andelman2009modulated,caplan2017colloquium}.

The LB framework is closely related to the Swift--Hohenberg equation \cite{swift1977hydrodynamic} and to Landau theory near a Lifshitz point \cite{goldenfeld1992lectures}. 
It provides a standard variational setting for analyzing periodic minima, energy-scaling laws, and pattern-selection mechanisms \cite{peletier2001spatial}. 
The inclusion of a cubic interaction term in the free-energy functional breaks the inversion symmetry of the order parameter, thereby enabling the model to capture asymmetric ordered phases and describe first-order phase transitions \cite{brazovskii1975phase,fredrickson1987fluctuation,zhang2014dynamic}. 
By characterizing the competition between local nonlinear interactions and long-range ordering effects, the LB framework predicts a rich variety of three-dimensional (3D) periodic structures. 
Beyond classical lamellar (LAM), hexagonally-packed cylinder (HEX), and body-centered cubic (BCC) phases \cite{podneks1996landau}, the model also admits topologically complex phases such as face-centered cubic (FCC), double gyroid (DG), and Frank--Kasper phases including A15 and $\sigma$ \cite{shi1999nature,mcclenagan2019landau}.
The numerous local minima of this energy landscape pose significant challenges for comprehensively exploring the configuration space and resolving these intricate spatial structures, thereby making it a rigorous testbed for advanced numerical methods.

Because of the strong nonlinearity of the LB energy functional and the presence of multiple local minima, the development of efficient numerical algorithms for computing stationary states remains a central challenge.
Existing numerical approaches can be broadly classified into three categories.
The first category considers directly solving the Euler--Lagrange (E-L) equations associated with the free-energy functional.
The second category comprises gradient flow approaches that describe the dynamical relaxation of ordered structures toward equilibrium states.
To enhance numerical stability and allow for large time steps, various temporal discretization strategies have been developed, including semi-implicit methods \cite{aland2012buckling}, exponential time differencing schemes \cite{du2021maximum}, and operator-splitting techniques \cite{lee2015first}.
To achieve nonlinear energy stability while retaining linear solvability at each time step, numerical techniques such as convex splitting schemes \cite{vignal2015energy,zhang2023spectral} and the scalar auxiliary variable (SAV) method \cite{shen2018scalar} have been systematically developed and widely applied in phase-field models.
Complementary to temporal discretization, efficient spatial schemes are equally essential. 
Spectral and pseudospectral methods \cite{chen1998applications} are particularly effective for periodic systems such as diblock copolymers, while finite difference \cite{hu2009stable} and finite element methods \cite{du2008adaptive} are also widely employed.
In contrast to dynamical evolution approaches, the third category formulates the problem as a direct optimization task, seeking stationary states through minimization of the discretized free-energy functional.
Recent studies indicate that optimization-based algorithms can exhibit superior computational efficiency compared with traditional gradient flow methods when computing stationary states in phase-field models \cite{jiang2020efficient}.
Representative examples include the Broyden--Fletcher--Goldfarb--Shanno (BFGS) algorithm applied to spherical harmonics discretizations \cite{luo2018phase}, as well as the adaptive accelerated Bregman proximal gradient (AA-BPG) method \cite{jiang2020efficient,bao2022adaptive}, which accelerates convergence in phase-field crystal models through optimization-based acceleration strategies.
The numerous local minima in this energy landscape pose significant challenges for systematically exploring the configuration space and resolving intricate spatial structures, thereby rendering the LB model a stringent testbed for advanced numerical optimization methods.

Despite the efficiency of these algorithms, a fundamental physical challenge comes from the spatial discretization. 
To simulate infinite-space LB theory, traditional numerical methods mainly rely on the Crystalline Approximant Method (CAM) \cite{zhang2008efficient,jiang2025approximation}, which restricts the system to a finite computational domain $\Omega$ with periodic boundary conditions, treating the structure as a periodic crystal. 
However, this creates a critical constraint: the ordered phases in the LB model have a natural characteristic length scale. 
If the size and shape of the computational domain do not strictly match this natural periodicity, the system is forced into a stressed state. 
Consequently, to remove this artificial stress and find the true equilibrium, the domain geometry cannot be treated as a fixed background; instead, it must be optimized together with the order parameter. 
To solve this coupled problem, traditional strategies typically use an alternating iterative scheme: optimizing the density field while fixing the domain $\Omega$, and then adjusting $\Omega$ based on the order parameter \cite{zhang2008efficient,jiang2013discovery}. 
This separated process makes the solver highly sensitive to the initial configuration. 
Without precise prior knowledge to construct a good initial guess, these solvers may easily get trapped in metastable states, hindering the discovery of complex or unknown ordered structures.

Physics-informed deep learning techniques, including Physics-Informed Neural Networks (PINNs) \cite{raissi2019physics} and the Deep Ritz Method (DRM) \cite{yu2018deep}, have recently emerged as effective frameworks for approximating solutions to partial differential equations and variational problems. 
However, classical PDE-based approaches like PINNs are often limited when discovering ordered structures in pattern-forming systems, because PINNs minimize the residual of the E-L equation, which only guarantees convergence to a stationary point. 
Since local minima, local maxima, and saddle points all satisfy the E-L equation, a PINN-based approach lacks the intrinsic driving force to distinguish true energy minimizers from high-energy saddle points.
Among these approaches, variational formulations such as the DRM are particularly well suited for energy minimization problems. 
By parameterizing the order parameter with a continuous neural network, these methods reduce the infinite-dimensional functional minimization problem to a finite-dimensional optimization problem over network parameters.
More importantly, this neural representation permits the domain geometry $\Omega$ to be incorporated as trainable variables instead of being prescribed before optimization.
However, directly applying the standard DRM to the LB model encounters a fundamental difficulty arising from the optimization dynamics of deep neural networks. 
Neural networks are typically initialized with small random weights to facilitate stable training \cite{glorot2010understanding,he2015delving}, while in the LB free-energy landscape, a near-zero order parameter corresponds to the disordered phase (a trivial solution), which can act as a strong local attractor.
Consequently, when initialized with small weights, the network is prone to converge to the disordered phase during the early stages of training, thereby failing to explore nontrivial ordered structures. 
Furthermore, when the disordered phase is unstable, the standard DRM often fails to explore the energy landscape sufficiently and identify multiple ordered phases. 
Without additional regularization, the optimization often becomes biased toward a certain metastable configuration, leaving other physically relevant solutions undiscovered.

To address these challenges, we propose a Geometry-Adaptive Deep Variational Framework (GeoDVF) to discover complex ordered structures.
The central component of this framework is a neural network representation of the order parameter field, and we treat the geometric parameters of the computational domain as learnable variables and embed them, together with the network parameters, into a unified computational graph.
This architecture enables the simultaneous variational optimization of both the physical field and the domain geometry through gradient-based training. 
To mitigate the strong attraction of the disordered phase induced by small-weight initialization and enlarge the range of attainable metastable states, we introduce a warmup penalty strategy.
By incorporating an additional penalty term into the loss function, the disordered phase is destabilized during the early stages of training, thereby driving the optimization trajectory toward ordered structures.
Numerical experiments demonstrate that GeoDVF can discover complex 3D phases, including BCC, FCC, and A15, starting from random initializations.
Furthermore, our method successfully identifies both stable and metastable states without requiring any prior structural information.

The remainder of this paper is organized as follows. 
Section~\ref{sec:lbmodel} introduces the LB model and reviews traditional numerical methods, highlighting how the alternating iterative strategies lead to strong dependence on initializations.
Section~\ref{sec:method} presents the proposed framework, including the network architecture and the training strategy, and provides a theoretical analysis showing that the warmup penalty destabilizes the disordered phase.
Section~\ref{sec:experiments} presents extensive numerical studies that assess the accuracy, robustness, and discovery capability of the proposed framework.
Section~\ref{sec:final} concludes the paper and discusses perspectives for further development.

\section{Landau--Brazovskii model}\label{sec:lbmodel}
We focus on the nondimensionalized form of the LB free energy, which describes the phase behavior of systems undergoing short-wavelength ordering. 
The system is characterized by a scalar order parameter $\phi(\bm{x})$, representing the local density variation. 
By rescaling the length and energy scales to eliminate redundant parameters, the effective free-energy functional is given by \cite{shi1999nature}
\begin{equation}\label{eq:lb_functional}
\mathcal{E}(\phi)=\int\left(\frac{1}{2}|(\Delta+1)\phi|^2+\frac{\tau}{2}\phi^2-\frac{\gamma}{6}\phi^3+\frac{1}{24}\phi^4\right)\,\mathrm{d}\bm{x},
\end{equation}
subject to the mass-conservation constraint
\begin{equation}
\label{eq:mass_constraint}
\int\phi\,\mathrm{d}\bm{x}=0.
\end{equation}
Here, $\tau$ represents the reduced temperature, and $\gamma$ characterizes the asymmetry of the ordered phases, which is essential for describing first-order phase transitions.

Formally, the scalar order parameter $\phi$ is defined on the entire space $\mathbb{R}^d$ ($d=2,3$). 
To enable numerical implementations, we utilize the CAM \cite{zhang2008efficient, jiang2025approximation} and approximate different phases with periodic crystal structures on finite domains. 
Specifically, we consider the order parameter $\phi$ on a box domain $\Omega$ with periodic boundary conditions to characterize each crystalline phase. 
Consequently, our analysis focuses on the free-energy density $f(\phi)$, defined as the energy per unit volume
\begin{equation}\label{eq:energy_density}
f(\phi)=\frac{1}{|\Omega|}\int_{\Omega}\left(\frac{1}{2}|(\Delta+1)\phi|^2+\frac{\tau}{2}\phi^2-\frac{\gamma}{6}\phi^3+\frac{1}{24}\phi^4\right)\,\mathrm{d}\bm{x},
\end{equation}
where $|\Omega|$ denotes the volume of the domain. 
Thus, the problem of finding local minima is reduced to minimizing this free-energy density $f$ on a finite periodic domain $\Omega$.

We introduce some notations to facilitate the subsequent analysis. 
The mass conservation \eqref{eq:mass_constraint} implies that the order parameter $\phi$ should be constrained to the manifold $\mathcal{M}$,
\begin{equation}\label{eq:feasible_manifold}
\phi\in\mathcal{M}=
\left\{\psi\in H_{\text{per}}^2(\Omega)\,\middle|\,\int_{\Omega}\psi\,\mathrm{d}\bm{x}=0\right\},
\end{equation}
where $H_{\text{per}}^2(\Omega)$ denotes the Sobolev space of periodic functions with square-integrable second derivatives. 
To enforce this constraint, we introduce the projection operator $\mathcal{P}$, which maps a function onto the zero-mean subspace
\begin{equation}\label{eq:projection}
\mathcal{P}\psi=\psi-\frac{1}{|\Omega|}\int_{\Omega}\psi\,\mathrm{d}\bm{x}.
\end{equation}
Our objective is to find the equilibrium configuration $\phi^*$ that minimizes the LB free-energy density on this manifold
\begin{equation*}
\phi^*=\mathop{\arg\min}_{\phi\in\mathcal{M}}\,f(\phi).
\end{equation*}

To explicitly incorporate the domain geometry into the variational framework, we map the physical domain to a fixed reference domain. 
The computational domain $\Omega$ is assumed to be an orthogonal box in $\mathbb{R}^d$ defined by $\Omega = \prod_{i=1}^d [0, L_i]$. 
We introduce the normalized coordinates $\tilde{x}_i = x_i / L_i$, which map $\Omega$ to the unit hypercube $D=[0,1]^d$. 
Accordingly, the order parameter on the fixed unit domain is $\tilde{\phi}(\tilde{\bm{x}}) = \phi(\bm{x})$. 
In the following discussions, we drop the tilde for brevity. 
It is important to note that the Jacobian determinant of this transformation, $J = \prod L_i = |\Omega|$, exactly cancels the volume normalization factor $1/|\Omega|$ in the free-energy density \eqref{eq:energy_density}. 
Consequently, the rescaled free-energy density is an integral over the fixed unit domain, with the geometric dependence transferred explicitly into the differential operator,
\begin{equation}\label{eq:rescaled_energy}
f(\phi, \bm{\omega}) = \int_{D} 
\left( \frac{1}{2} |(\Delta_{\bm{\omega}} + 1)\phi|^2 + \frac{\tau}{2}\phi^2 - \frac{\gamma}{6}\phi^3 + \frac{1}{24}\phi^4 \right) 
\,\mathrm{d}\bm{x},
\end{equation}
where $\bm{\omega} = (L_1, \dots, L_d)$ represents the geometric parameters. 
The geometry-adapted Laplacian $\Delta_{\bm{\omega}}$ is defined as
\begin{equation*}
\Delta_{\bm{\omega}} = \sum_{i=1}^d \frac{1}{L_i^2} \cdot \frac{\partial^2}{\partial x_i^2}.
\end{equation*}
Specifically, for a two-dimensional (2D) domain, $\bm{\omega}=(L_1, L_2)$ and $\Delta_{\bm{\omega}} = L_1^{-2}\partial_{xx} + L_2^{-2}\partial_{yy}$; 
for a 3D domain, $\bm{\omega}=(L_1, L_2, L_3)$ and $\Delta_{\bm{\omega}} = L_1^{-2}\partial_{xx} + L_2^{-2}\partial_{yy} + L_3^{-2}\partial_{zz}$. 
This reformulation reveals that the domain parameters $\bm{\omega}$ act as coefficients within the differential operator. 
Therefore, finding the equilibrium state corresponds to minimizing $f(\phi, \bm{\omega})$ simultaneously with respect to the order parameter $\phi$ (on the fixed domain $D$) and the geometry parameters $\bm{\omega}$.

For each ordered structure, the free-energy density $f$ explicitly depends on the domain geometry $\bm{\omega}$. 
Consequently, the computational domain must be adjusted adaptively to relax the stress and minimize the free-energy density. 
The optimization process typically follows an alternating iterative scheme \cite{jiang2013discovery}. 
\begin{enumerate}[label=\textbf{Step \arabic*}, leftmargin=*]
\item With fixed domain parameters $\bm{\omega}^k$, evolve the order parameter $\phi$ to a local minimum:
\begin{equation}\label{eq:fix_domain}
\phi^{k+1}=\mathop{\arg\min}_{\phi\in\mathcal{M}}\,f(\phi,\bm{\omega}^k).
\end{equation}
\item With the fixed profile $\phi^{k+1}$, optimize the domain parameters $\bm{\omega}$ to minimize the free-energy density
\begin{equation}\label{eq:fix_phi}
\bm{\omega}^{k+1}=\mathop{\arg\min}_{\bm{\omega}}\,f(\phi^{k+1},\bm{\omega}).
\end{equation}
\item Repeat Step 1 and Step 2 until the error falls below a prescribed tolerance.
\end{enumerate}

To efficiently solve the first subproblem \eqref{eq:fix_domain}, one can employ the Allen-Cahn gradient flow with mass conservation 
\begin{equation*}
\frac{\partial \phi}{\partial t} = -\mathcal{P} \frac{\delta f}{\delta \phi} = - \mathcal{L}_{\bm{\omega}^k}\phi - \mathcal{P}g(\phi),
\end{equation*}
where the linear operator $\mathcal{L}_{\bm{\omega}}$ and a nonlinear term $g(\phi)$ are defined as 
\begin{equation*}
\mathcal{L}_{\bm{\omega}} = (\Delta_{\bm{\omega}}+1)^2 + \tau,\quad g(\phi)=-\frac{\gamma}{2}\phi^2+\frac{1}{6}\phi^3.
\end{equation*}
For temporal discretization, a semi-implicit scheme can be adopted to ensure stability with large time step sizes \cite{chen1998applications}
\begin{equation}\label{eq:semi_implicit}
\frac{\phi^{n+1}-\phi^n}{\Delta t} = - \mathcal{L}_{\bm{\omega}^k}\phi^{n+1} - \mathcal{P}g(\phi^n).
\end{equation}
Let $\mathcal{F}$ and $\mathcal{F}^{-1}$ denote the discrete Fourier transform and its inverse, respectively, and denote the Fourier coefficients as $\hat{\phi}(\bm{k}) = \mathcal{F}(\phi)(\bm{k})$. 
In the Fourier space, the differential operator $\mathcal{L}_{\bm{\omega}^k}$ is diagonalized, so applying $\mathcal{F}$ to \eqref{eq:semi_implicit} allows for a fast pointwise update based on pseudospectral methods:
\begin{equation*}
\hat{\phi}^{n+1}(\bm{k}) = \frac{\hat{\phi}^n(\bm{k}) - \Delta t \cdot \mathcal{F}(\mathcal{P}g(\phi^n))(\bm{k})}{1 + \Delta t \left[ \left(1 - |\bm{k}|_{\bm{\omega}^k}^2\right)^2 + \tau \right]}\, ,
\end{equation*}
where the geometry-dependent wave vector magnitude is $|\bm{k}|_{\bm{\omega}}^2 = \sum_i (2\pi k_i / L_i)^2$.
The solution to \eqref{eq:fix_domain} should satisfy the E-L equation $\mathcal{L}_{\bm{\omega}^k}\phi^{n+1} + \mathcal{P}g(\phi^{n+1}) = 0$.

For the second subproblem \eqref{eq:fix_phi}, the domain parameters $\bm{\omega}=(L_1, \dots, L_d)$ are typically updated via gradient descent. 
While the minimization of $f$ with respect to $\omega$ theoretically admits an analytic solution for a fixed $\phi$, directly imposing this exact geometric solution at each step may cause numerical instability. 
The exact solution derived from a transient non-equilibrium $\phi$, particularly during early iterations, can force the domain to change rapidly, leading to oscillation or failure. 
Therefore, gradient-based updates act as a  relaxation mechanism to ensure that the domain geometry evolves smoothly together with the order parameter $\phi$. 
It is important to note that the bulk energy terms are invariant under domain scaling when $\phi$ is fixed on the reference domain $D$, so the gradient arises solely from the elastic energy term. 
By Parseval's identity, the gradient with respect to the $i$-th domain length $L_i$ ($i=1,\dots,d$) is derived analytically as
\begin{equation*}
\frac{\partial f}{\partial L_i} = \frac{2}{L_i} \sum_{\bm{k}} \left( 1 - |\bm{k}|_{\bm{\omega}}^2 \right) \left( \frac{2\pi k_i}{L_i} \right)^2 |\hat{\phi}(\bm{k})|^2,
\end{equation*}
where $k_i$ denotes the $i$-th component of the wave vector $\bm{k}$.

This alternating strategy highlights the fundamental dependency on initialization: the solver requires not only an initial configuration $\phi^0$ when solving \eqref{eq:fix_domain}, but also a compatible initial domain size $\bm{\omega}^0$. 
We denote this coupled deterministic solver process as
\begin{equation}\label{eq:solver}
(\phi^*,\bm{\omega}^*)=\text{Solver}(\phi^0,\bm{\omega}^0).
\end{equation}
To explicitly demonstrate the sensitivity of the alternating iterative scheme to initial guesses $(\phi^0, \bm{\omega}^0)$, we conduct a series of numerical experiments on the 2D LB model.
Apart from the disordered phase, which is the trivial solution $\phi=0$, the 2D LB model admits two primary periodic ordered structures: LAM and HEX, as visualized in Figure~\ref{fig:LB2d}. 
Theoretically, both the LAM and HEX phases can share the same optimal computational domain $\Omega_{\text{opt}}=[0, 8\pi] \times [0, 16\pi / \sqrt{3}]$, i.e., $\bm{\omega}_{\text{opt}}=(8\pi, 16\pi / \sqrt{3})$. 
We utilize the standard initialization to construct the initial profiles for these phases \cite{wickham2003nucleation}, 
\begin{equation*}
\phi^0(\bm{x}) = a_1 \cos(\bm{k}_1 \cdot \bm{x}) + a_2 (\cos(\bm{k}_2 \cdot \bm{x}) + \cos(\bm{k}_3 \cdot \bm{x})),
\end{equation*}
where the wave vectors are $\bm{k}_1=(1,0)$, $\bm{k}_2=(-1/2, \sqrt{3}/2)$, and $\bm{k}_3=(-1/2, -\sqrt{3}/2)$. 
We set $(a_1, a_2)=(1, 0)$ for LAM ($\phi^0_{\text{LAM}}$) and $(a_1, a_2)=(1, 1)$ for HEX ($\phi^0_{\text{HEX}}$).

\begin{figure}[!t]
\centering
\makesubfig{fig:LB2d_a}
\makesubfig{fig:LB2d_b}
\makesubfig{fig:LB2d_c}
\makesubfig{fig:LB2d_d}
\includegraphics[width=\linewidth]{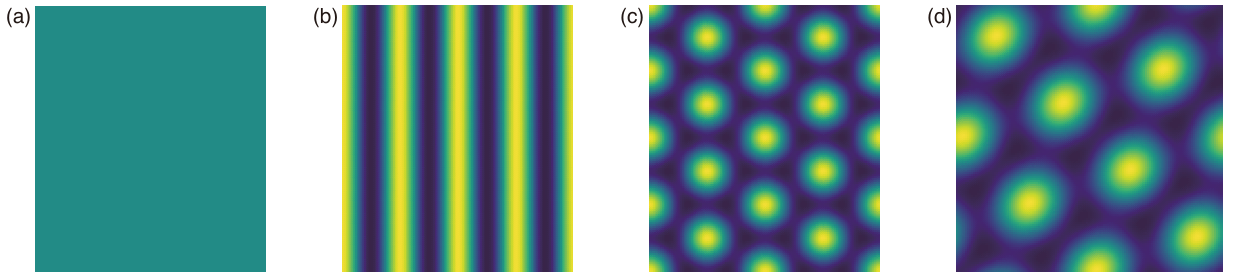}
\caption{Representative stationary states of the 2D LB model. 
(a) The disordered phase. 
(b) LAM phase. 
(c) HEX phase obtained within an optimal computational domain. 
(d) A distorted hexagonal phase (denoted as HEX*). This is a high-energy metastable state where the lattice structure is deformed due to the artificial stress caused by domain mismatch.}
\label{fig:LB2d}
\end{figure}

We consider two parameter sets: $(\tau, \gamma)=(0.001, 0.6)$, where HEX is the global minimum and the disordered phase is metastable; and $(\tau, \gamma)=(-0.4, 0.3)$, where LAM is the global minimum. 
To test the robustness of the solver defined in \eqref{eq:solver}, we introduce a non-optimal square domain $\Omega_{\text{s}}=[0, 20]^2$ (i.e., $\bm{\omega}_{\text{s}}=(20, 20)$) and a random initialization $\phi^0_{\text{randn}}(\bm{x}) \sim \mathcal{N}(0, 1)$ sampled independently at each grid point. 
We evaluate the solver's performance using various combinations of initial fields and domain geometries, and the convergence results are summarized in Table~\ref{tab:sensitivity}. 
The experiments reveal a significant dependency on the initialization. 
Even when initialized with the correct topological structure ($\phi^0_{\text{HEX}}$), starting with a non-optimal domain $\bm{\omega}_{\text{s}}$ can prevent the alternating solver from reaching the stress-free ground state. 
Instead, it gets trapped in a distorted hexagonal phase (denoted as HEX*, see Figure~\figsubref{fig:LB2d}{fig:LB2d_d}), which is a metastable state with higher energy caused by the domain mismatch stress. 
When starting from random noise $\phi^0_{\text{randn}}$, the solver often fails to identify the global minimum if the domain size is not pre-tuned. 
For instance, in the HEX-stable regime with $\bm{\omega}_{\text{s}}$, the system collapses into the disordered phase rather than evolving into ordered structures. 
These results highlight that without precise prior knowledge of both the pattern periodicity and the domain size, the traditional alternating iterative strategy may struggle to identify the global energy minimum.

\begin{table}[!t]
\centering
\begin{threeparttable}
\caption{Convergence outcomes for the alternating iterative solver under different parameter sets $(\tau, \gamma)$ and initial guesses.}
\label{tab:sensitivity}
\begin{tabular}{ccccccc}
\toprule
\multirow{2}{*}{$(\tau,\gamma)$} & \multicolumn{6}{c}{Initial guess $(\phi^0,\omega^0)$} \\ \cmidrule(lr){2-7}
 & $(\phi^0_{\text{HEX}},\omega_{\text{opt}})$ & $(\phi^0_{\text{HEX}},\omega_{\text{s}})$ & $(\phi^0_{\text{LAM}},\omega_{\text{opt}})$ & $(\phi^0_{\text{LAM}},\omega_{\text{s}})$ & $(\phi^0_{\text{randn}},\omega_{\text{opt}})$ & $(\phi^0_{\text{randn}},\omega_{\text{s}})$ \\
\midrule
(0.001, 0.6) & HEX & HEX* & disorder & disorder & HEX & disorder \\
(-0.4, 0.3) & HEX & LAM & LAM & LAM & LAM & LAM \\
\bottomrule
\end{tabular}
\end{threeparttable}
\end{table}

\section{Geometry-Adaptive Deep Variational Framework}\label{sec:method} 
To mitigate the heavy dependence on the initial guess inherent in the alternating iteration method, we propose GeoDVF to optimize the order parameter $\phi$ and the geometry parameter $\bm{\omega}$ simultaneously. 
An overview of our proposed method is presented in Figure~\ref{fig:architecture}. 

Central to this framework, we employ a neural network to parameterize the order parameter $\phi$, denoted as $\phi_{\theta}$, while we treat the domain geometry as learnable parameters jointly optimized with the network weights $\theta$.
Specifically, for a $d$-dimensional box domain with side lengths $\bm{\omega}=(L_1, \dots, L_d)$, we introduce the inverse geometric parameters $\bm{\beta} = (\beta_1, \dots, \beta_d)$, defined element-wise as $\beta_i = 1/L_i$. 
This reparameterization offers two key advantages: 
first, it naturally maps the domain sizes $L_i$ (typically $>1$) to a bounded range $\beta_i \in (0, 1)$, which stabilize training; 
second, it simplifies the backpropagation process, as the free-energy density depends on the geometric parameters solely through inverse square terms (i.e., $1/L_i^2 = \beta_i^2$), transforming rational functions into polynomial terms with respect to the trainable variables. 
Consequently, the complete set of trainable parameters for the framework is defined as $\Theta = (\theta, \bm{\beta})$. 

\begin{figure}[!t]
\centering
\includegraphics[width=\linewidth]{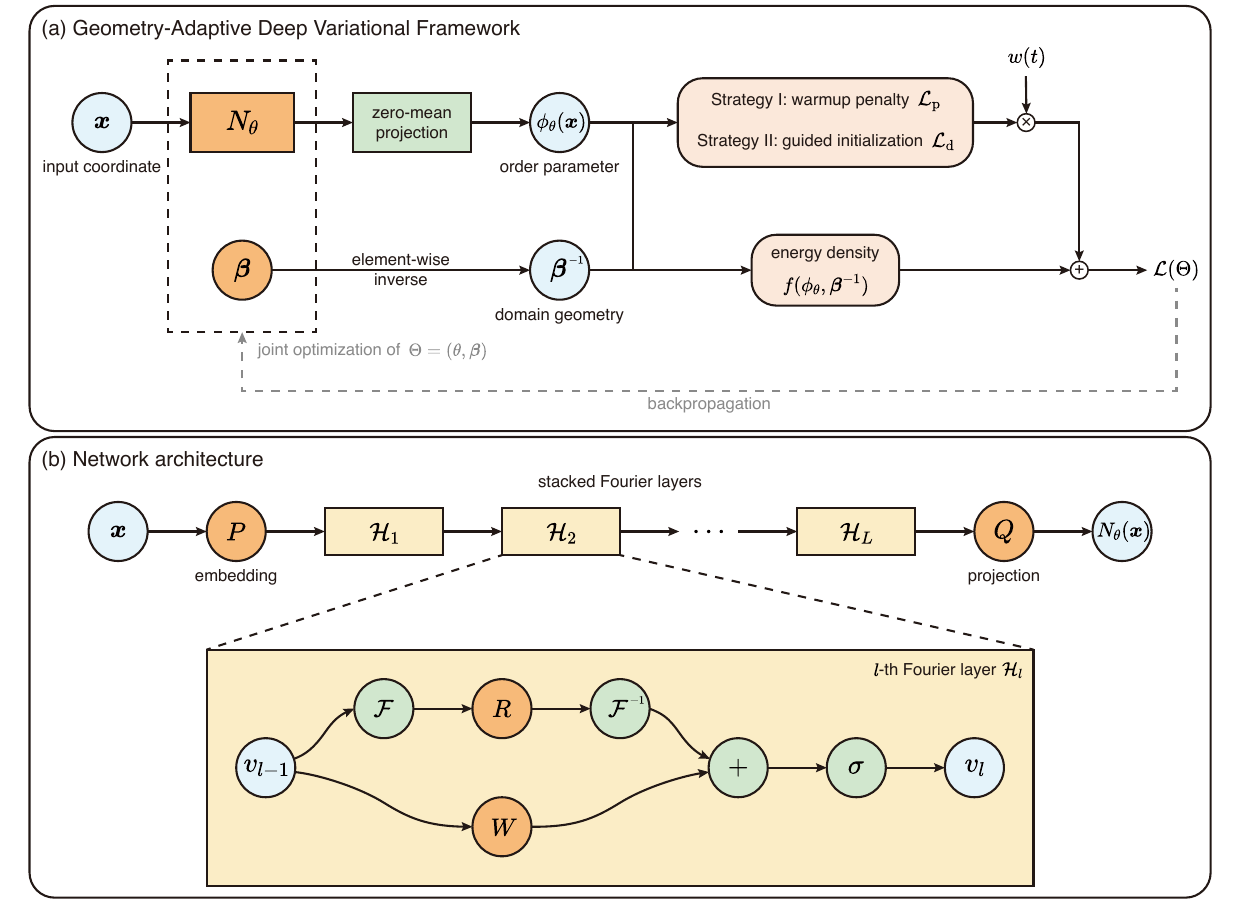}
\caption{An overview of GeoDVF. 
(a) The workflow of GeoDVF. The system takes spatial coordinates $\bm{x}$ as input and simultaneously optimizes the network parameters $\theta$ (for the order parameter field $\phi_{\theta}$) and the inverse geometric parameters $\bm{\beta}$ (for the domain geometry). The zero-mean projection enforces the mass conservation constraint, converting the raw network output $N_{\theta}(\bm{x})$ into the physical order parameter $\phi_{\theta}$. The total loss $\mathcal{L}(\Theta)$ integrates the physical free-energy density $f$ with auxiliary terms from the warmup penalty or guided initialization. The dashed line indicates that the unified parameter set $\Theta = (\theta, \bm{\beta})$ is updated jointly via backpropagation.
(b) Detailed architecture of the neural network $N_{\theta}$. The coordinate embedding layer $P$ maps the input coordinates $\bm{x}$ directly into a high-dimensional feature function. This representation evolves through a stack of Fourier layers $\mathcal{H}_l$ and is mapped back to the unconstrained scalar field $N_{\theta}$ by the projection layer $Q$. The inset illustrates the operations within a single Fourier layer.}
\label{fig:architecture}
\end{figure}

\subsection{Neural network architecture}
To effectively parameterize the spatially continuous order parameter $\phi(\bm{x})$, we adopt a neural network inspired by the Fourier Neural Operator (FNO) \cite{li2021fourier}. 
Unlike traditional convolutional neural networks \cite{o2015introduction} that rely on local operations on fixed grids, the Fourier architecture operates in the spectral domain, offering a natural representation for periodic structures and enabling resolution-independent optimization.

The proposed architecture comprises three sequential components: an embedding layer, stacked Fourier layers, and a projection layer. 
The Fourier and projection layers follow the standard FNO design. 
The lifting layer in a standard FNO is designed to embed a low-dimensional input function into a high-dimensional latent space, while our embedding layer $P$ acts as a coordinate embedding, mapping the spatial coordinates $\bm{x} \in D$ directly to a high-dimensional hidden feature field $v_0(\bm{x}) = P(\bm{x})$. 
Although both are implemented via a pointwise affine transformation, our embedding operation initializes the neural representation directly from the domain geometry itself, rather than from an external function.

The core building block of the subsequent deep processing is the Fourier layer. 
Let $v_{l-1}: D \to \mathbb{R}^{d_v}$ be the input feature field to the $l$-th layer and $v_{l}$ be the output. The update rule is defined as
\begin{equation*}
\label{eq:fourier_update}
v_{l}(\bm{x}) = \sigma \left( W (v_{l-1})(\bm{x}) + \mathcal{K}(v_{l-1})(\bm{x}) \right), \quad  \bm{x} \in D,
\end{equation*}
where $\sigma$ is a nonlinear activation function, and $W$ is a pointwise affine transformation. The key component is the global spectral convolution operator $\mathcal{K}$, defined by
\begin{equation*}
\label{eq:spectral_conv}
\mathcal{K}(v)(\bm{x}) = \mathcal{F}^{-1} \Bigl( R \cdot \mathcal{F}(v) \Bigr)(\bm{x}).
\end{equation*}
Here, $\mathcal{F}$ and $\mathcal{F}^{-1}$ denote the Fourier transform and its inverse, respectively. 
Following the standard FNO formulation, we implement an efficient spectral truncation by defining the set of active Fourier modes as the box region $Z_{k_{\max}} = \{k \in \mathbb{Z}^d : |k_j| \leqslant k_{\max, j}$, for $j=1,\dots,d\}$. 
The term $R$ is then parameterized as a learnable complex-valued tensor that performs linear transformations exclusively on the modes $k \in Z_{k_{\max}}$.

Finally, a projection layer $Q$, also implemented as a pointwise affine transformation, maps the high-dimensional features $v_L(\bm{x})$ of the last Fourier layer back to an unconstrained scalar field $N_{\theta}(\bm{x})$. 
To provide a unified representation of the network architecture, let us define the operator for the $l$-th Fourier layer as $\mathcal{H}_l(v) := \sigma \left( W v + \mathcal{K} v \right)$. 
The complete network parameterization can thus be expressed as the composition of the coordinate embedding, $L$ stacked Fourier layers, and the final projection:
\begin{equation}
\label{eq:nn_composition}
N_{\theta}(\bm{x}) = \left( Q \circ \mathcal{H}_L \circ \cdots \circ \mathcal{H}_1 \circ P \right)(\bm{x}),
\end{equation}
where $\circ$ denotes function composition. 
It serves as the structural precursor to the physical order parameter $\phi_{\theta}$, as discussed in the subsequent section.

This spectral parameterization is particularly well suited for minimizing the LB free energy due to its intrinsic mathematical structure. 
By operating directly in the frequency domain, the network achieves a global receptive field within a single layer. 
This enables the efficient capture of long-range spatial correlations and interactions that are essential for resolving complex phase topologies.
Furthermore, reliance on the Fourier basis automatically enforces periodic boundary conditions, ensuring that the generated field remains physically consistent with the assumptions of CAM without imposing additional constraints.
Finally, since the weight parameters $R$ are learned in the spectral domain, the resulting representation is continuous and mesh-invariant, which enables the model to be evaluated on computational grids of arbitrary resolution, thereby resolving fine interfacial structures without retraining.

\subsection{Mass-conservation constraint}
Relying on soft constraints, i.e., penalty terms in the loss function, introduces additional hyperparameters, which often leads to complex optimization landscapes and requires delicate tuning of penalty coefficients \cite{sun2020surrogate, wang2021understanding}. 
To address these challenges and ensure physical consistency, we enforce the constraint of the feasible manifold $\mathcal{M}$ in \eqref{eq:feasible_manifold} directly through the network architecture.
Recall that a valid solution $\phi \in \mathcal{M}$ must satisfy both periodic boundary conditions and mass conservation, and the periodic boundary condition is inherently guaranteed by our choice of the Fourier-based architecture. 
As detailed in the previous section, the network operates on the basis of trigonometric functions, which naturally imposes periodicity on the generated field without the need for padding strategies or auxiliary boundary loss terms.

To strictly enforce the mass-conservation constraint \eqref{eq:mass_constraint}, we employ a zero-mean projection strategy at the network output. 
Recalling the unconstrained network output $N_{\theta}(\bm{x})$ defined in \eqref{eq:nn_composition}, the final physical order parameter $\phi_{\theta}(\bm{x})$ is obtained by applying the mass-conservation projection $\mathcal{P}$ \eqref{eq:projection}
\begin{equation}\label{eq:hard_constraint}
\phi_{\theta}(\bm{x}) = \mathcal{P}N_{\theta}(\bm{x}) = N_{\theta}(\bm{x}) - \frac{1}{|\Omega|}\int_{\Omega} N_{\theta}(\bm{y}) \,\mathrm{d}\bm{y}.
\end{equation}
In the discrete setting, the integral is replaced by the spatial average. 
This operation ensures that the total integral of $\phi_{\theta}$ vanishes strictly throughout the training process, effectively constraining the optimization search space within $\mathcal{M}$.

\subsection{Loss function and training strategies}
\label{sec:method_loss}
Our training objective is based on the DRM \cite{yu2018deep}, which reformulates the variational problem as an optimization problem over the function space of neural networks. 
For the unified trainable parameters $\Theta = (\theta, \bm{\beta})$, the primary physical loss function is defined as
\begin{equation}
\label{eq:DRM_loss}
\mathcal{L}_{\text{DRM}}(\Theta) = f(\phi_{\theta}, \bm{\beta}^{-1}),
\end{equation}
where $\bm{\beta}^{-1}$ (element-wise inverse) reconstructs the physical domain lengths required for the energy computation in \eqref{eq:rescaled_energy}. 
A distinct advantage of our network architecture is its natural generation of solutions on uniform periodic grids, which facilitates the use of pseudospectral methods. 
By computing derivatives in the frequency domain, we achieve high accuracy while effectively avoiding the discretization errors of finite difference schemes or the computational cost of automatic differentiation \cite{baydin2018automatic,raissi2019physics} for high-order derivatives.

Beyond the DRM loss \eqref{eq:DRM_loss}, we introduce two specialized training strategies to address the challenges of nonconvex optimization and to target specific equilibrium configurations.

\subsubsection*{Strategy I: Warmup penalty for phase discovery}
When training from scratch with random initializations, the network parameters are typically small, resulting in an output $\phi_{\theta}$ close to zero. 
Because the disordered phase ($\phi=0$) is a local minimum for the reduced temperature $\tau>0$, the optimization can easily get trapped in the disordered phase.

To mitigate this, we propose a \emph{warmup penalty} strategy.
The central idea is to drive the order parameter to acquire energy at the physically critical length scale during the early stages of training. 
We construct a penalty term based on the projection of the order parameter onto the critical sphere in the reciprocal space. 
Consider the test functions in the physical space as $\psi_{\bm{n}}(\bm{x}) = \cos(\bm{n} \cdot \bm{x})$, where $\bm{n} \in \mathbb{S}^{d-1}$ is a unit vector. 
In the computational domain $\tilde{\bm{x}} \in D$, utilizing the inverse geometric parameter $\bm{\beta}$, the effective test function in the reference domain $D$ becomes $\tilde{\psi}_{\bm{n}}(\tilde{\bm{x}}) = \cos(\bm{n} \cdot (\bm{\beta}^{-1} \odot \tilde{\bm{x}}))$, where $\odot$ denotes elementwise multiplication. 
The penalty term is defined as
\begin{equation*}
\mathcal{L}_{\text{p}}(\Theta) = \left( \;\bbint_{\mathbb{S}^{d-1}} \left| \int_{D} \phi_{\theta}(\tilde{\bm{x}}) \, \tilde{\psi}_{\bm{n}}(\tilde{\bm{x}}) \, \mathrm{d}\tilde{\bm{x}} \right|\, \mathrm{d}\bm{n} - C \right)^2,\quad d=2,3,
\end{equation*}
where $\bbint_{\;\,\mathbb{S}^{d-1}}$ denotes the averaged integral over the unit sphere, and $C > 0$ is a hyperparameter. 
This term penalizes the deviation of the order parameter's projection on the critical sphere. 
The time-dependent total loss for this strategy is given by 
\begin{equation*}
\mathcal{L}(\Theta,t) = \mathcal{L}_{\text{DRM}}(\Theta) + w_{\text{p}}(t) \mathcal{L}_{\text{p}}(\Theta),
\end{equation*}
where $w_{\text{p}}(t)$ is a scheduler function. 
We set $w_{\text{p}}(t) > 0$ for the initial $T_{\text{warm}}$ steps to seed the correct mode structure and drive the geometric parameters toward a stress-free configuration. 
Subsequently, we set $w_{\text{p}}(t) = 0$ for the remainder of the training to allow the physics to dominate the relaxation process.

\subsubsection*{Strategy II: Guided initialization for targeted search}
Although the warmup penalty facilitates the spontaneous discovery of various minima, this exploration may not be exhaustive, as certain topologically complex phases can reside in extremely narrow basins of attraction that are inaccessible from random initializations. 
This limitation parallels traditional numerical iterative schemes, which typically require carefully constructed initial guesses to converge to specific target phases.

To deal with this issue, we employ a guided initialization strategy.
Similar to imposing an initial condition in classical solvers, this strategy serves to map an order-parameter initialization directly into the network's parameter space by introducing a data-fidelity term $\mathcal{L}_{\text{d}}$ as 
\begin{equation*}
\mathcal{L}_{\text{d}}(\Theta) = \|\phi_{\theta} - \phi_{\text{init}}\|^2_{L^2},
\end{equation*}
where $\phi_{\text{init}}$ is a pre-calculated specific initial guess. 
Then, the total loss for this strategy is given by 
\begin{equation*}
\mathcal{L}(\Theta,t) = \mathcal{L}_{\text{DRM}}(\Theta) + w_{\text{d}}(t) \mathcal{L}_{\text{d}}(\Theta).
\end{equation*}
Similar to the warmup strategy, the weighting coefficient $w_{\text{d}}(t)$ acts as a time-dependent scheduler. 
We maintain $w_{\text{d}}(t) > 0$ during the initial $T_{\text{guide}}$ steps to drive the optimization trajectory towards the target phase $\phi_{\text{init}}$,.
In the subsequent training stage, we set $w_{\text{d}}(t) = 0$ to release the constraint on data matching, allowing the network to fine-tune the solution and converge solely according to the physical free energy, thereby ensuring the final state is physically accurate.

Additionally, the successful stabilization of these phases after removing the guidance provides strong evidence of the expressivity of our proposed method. 
This observation indicates that the failure to discover these phases spontaneously arises from the nonconvex optimization landscape, rather than from any intrinsic limitation of the network architecture. 
Our numerical experiments further confirm that, once initialized within the appropriate basin of attraction, the neural network has sufficient representational capacity to accurately capture and stably preserve these intricate topological structures.

\subsection{Theoretical analysis of the warmup penalty}
\label{sec:analysis}
To theoretically validate the warmup penalty strategy introduced in Section~\ref{sec:method_loss}, we analyze the local stability of the disordered phase $\phi=0$ under the modified loss landscape. 
We formally establish that although the disordered phase constitutes a local minimum of the free energy, the incorporation of the penalty term destabilizes it by inducing a strict descent direction.
For clarity, let us define the projection intensity functional $\mathcal{S}(\phi)$ and the total penalized energy density $f_{\text{p}}(\phi)$ as
\begin{equation*}
\mathcal{S}(\phi)=\bbint_{\mathbb{S}^{d-1}}\left|\frac{1}{|\Omega|}\int_{\Omega}\phi(\bm{x})\cdot\psi_{\bm{n}}(\bm{x})\,\mathrm{d}\bm{x}\right|\,\mathrm{d}\bm{n}, \quad f_{\text{p}}(\phi)=f(\phi)+w_{\text{p}}(\mathcal{S}(\phi)-C)^2,
\end{equation*}
where $\psi_{\bm{n}}(\bm{x})$ are the test functions on the unit spectral shell defined in the previous subsection.
We present the following theorem to demonstrate this destabilization effect.

\begin{theorem}[Local instability of the disordered phase]
\label{thm:instability}
Consider the Landau--Brazovskii free-energy functional $f(\phi)$ with $\tau > 0$ subject to the mass-conservation constraint, and then the disordered phase $\phi = 0$ is a local minimum of the free energy $f(\phi)$.
For $w_{\text{p}} > 0$ and $C > 0$, the disordered phase is strictly unstable for the penalized energy $f_{\text{p}}(\phi)$. 
Specifically, any mass-conserving direction $\phi_0$ with a non-vanishing projection intensity (i.e., $\mathcal{S}(\phi_0) > 0$) constitutes a descent direction.
\end{theorem}

\begin{proof}
To analyze the stability, we compute the one-sided Gateaux derivative of the functional at $\phi=0$ along a perturbation direction $\phi_0$. 
Define the directional derivative of a functional $J(\phi)$ as
\begin{equation*}
\delta J(0; \phi_0) = \lim_{\varepsilon \to 0^+} \frac{J(\varepsilon \phi_0) - J(0)}{\varepsilon}.
\end{equation*}
Crucially, to respect the mass-conservation constraint, the perturbation direction must lie in the tangent space of the manifold $\mathcal{M}$ at $\phi=0$, which requires $\phi_0$ to be a zero-mean function.
The free-energy density functional $f(\phi)$ is Fr\'echet differentiable, and thus Gateaux differentiable on the Sobolev space $H^2(\Omega)$. 
$\phi = 0$ is a local minimum of the physical energy $f(\phi)$ along the subspace $\mathcal{M}$ for $\tau > 0$, and $\delta f(0; \phi_0)$ vanishes identically for any admissible perturbation $\phi_0$. 
Therefore, the energy scales quadratically near $\phi=0$, i.e., $f(\varepsilon\phi_0) - f(0) = \mathcal{O}(\varepsilon^2)$.

Now consider the penalty term $J_{\text{p}}(\phi) = w_{\text{p}}(\mathcal{S}(\phi)-C)^2$. 
The projection intensity $\mathcal{S}(\phi)$ satisfies positive homogeneity: $\mathcal{S}(\varepsilon\phi_0) = \varepsilon\mathcal{S}(\phi_0)$ for $\varepsilon > 0$. 
Let $A = \mathcal{S}(\phi_0) \geqslant 0$ be the projection intensity of the perturbation, and the value of the penalty term at the perturbed state is
\begin{equation*}
J_{\text{p}}(\varepsilon\phi_0) = w_{\text{p}}(\varepsilon A - C)^2 = w_{\text{p}}(C^2 - 2AC\varepsilon + A^2\varepsilon^2).
\end{equation*}
At the disordered phase $\phi=0$, we have $\mathcal{S}(0)=0$ and $J_{\text{p}}(0) = w_{\text{p}}C^2$, so the directional derivative is calculated as
\begin{equation*}
\delta J_{\text{p}}(0; \phi_0) = \lim_{\varepsilon \to 0^+} \frac{w_{\text{p}}(C^2 - 2AC\varepsilon + A^2\varepsilon^2) - w_{\text{p}}C^2}{\varepsilon} = -2 w_{\text{p}} A C.
\end{equation*}
The directional derivative of the total penalized energy is the sum of the physical and penalty contributions
\begin{equation*}
\delta f_{\text{p}}(0; \phi_0) = \delta f(0; \phi_0) + \delta J_{\text{p}}(0; \phi_0) = 0 - 2 w_{\text{p}} A C = - 2 w_{\text{p}} A C.
\end{equation*}
This result implies that if there exists a direction $\phi_0\in\mathcal{M}$ such that $\mathcal{S}(\phi_0) > 0$, the training loss will strictly decrease along this direction when $w_{\text{p}}, C > 0$.

To complete the proof, we explicitly construct a fundamental mode $\phi_0(\bm{x})=\cos(2\pi x_1 / L_1)$ aligned with the first spatial dimension, which is periodic and satisfies the mass conservation. 
We now verify that the projection intensity $\mathcal{S}(\phi_0)$ is strictly positive. 
For a unit vector $\bm{n}=(n_1,\dots,n_d)$, we evaluate the projection integral
\begin{equation*}
\mathcal{I}(\bm{n}) = \frac{1}{|\Omega|}\int_{\Omega}\phi_0(\bm{x})\cdot\psi_{\bm{n}}(\bm{x})\,\mathrm{d}\bm{x} = \frac{2n_1L_1}{n_1^2L_1^2-4\pi^2} \cdot \cos\left(\frac{1}{2}\bm{n}\cdot\bm{\omega}\right) \cdot \prod_{i=1}^d \frac{\sin(n_iL_i / 2)}{n_iL_i / 2},
\end{equation*}
and we consider a generic $\bm{n}$ such that the denominator is nonzero.
Observe that $\mathcal{I}(\bm{n})$ is a real-analytic function of $\bm{n}$ restricted to the unit sphere $\mathbb{S}^{d-1}$, so its zero set is determined by the roots of the sine and cosine factors (e.g., $n_i L_i / 2 = k\pi$). 
Geometrically, these conditions define a union of lower-dimensional hypersurfaces rather than the entire sphere, so the set of directions where $\mathcal{I}(\bm{n}) = 0$ has measure zero on $\mathbb{S}^{d-1}$.
This implies that $|\mathcal{I}(\bm{n})| > 0$ almost everywhere, so the total projection intensity 
\begin{equation*}
\mathcal{S}(\phi_0) = \bbint_{\mathbb{S}^{d-1}} |\mathcal{I}(\bm{n})| \,\mathrm{d}\bm{n}
\end{equation*}
is strictly positive.
Finally, we obtain $\delta f_{\text{p}}(0; \phi_0) < 0$, and the existence of a negative directional derivative proves that $\phi=0$ is not a local minimum in the penalized loss landscape.
\end{proof}

\begin{remark}
The instability arises from the difference in scaling behaviors near the disordered phase. 
The physical energy $f(\phi)$ scales quadratically ($\sim \varepsilon^2$) due to variational stability, whereas the penalty term introduces a negative linear scaling ($\sim -2w_{\text{p}}AC\varepsilon$) due to the penalty term. 
For sufficiently small $\varepsilon$, the linear term dominates, effectively creating a slope that pushes the optimization trajectory away from the disordered phase.
\end{remark}

\begin{remark}
It is worth noting that the absolute value operation in the definition of $\mathcal{S}(\phi)$ introduces a singularity of the gradient at the disordered phase. 
In classical linear instability, the driving force vanishes as $\phi \to 0$, while the singularity here implies a discontinuous jump in the derivative, resulting in a nonvanishing, constant driving force for infinitesimal perturbations. 
In numerical implementation, since $\phi$ is never strictly zero, this property allows the system to escape the zero attractor significantly faster.
\end{remark}

\section{Numerical experiments}\label{sec:experiments}
In this section, we conduct systematic numerical experiments to validate the effectiveness, accuracy, and exploration capability of our proposed method.

\paragraph{Implementation details} Unless otherwise stated, the neural network employed in our experiments is constructed with four stacked Fourier layers with $64$ hidden channels. 
We adopt GELU \cite{hendrycks2016gaussian} as the nonlinear activation function $\sigma$.
The frequency spectrum is truncated in each layer by retaining the lowest $k_{\max}=16$ modes along each spatial dimension. 
The network is implemented in PyTorch and trained using the AdamW optimizer \cite{loshchilov2019decoupled}. 
The initial learning rate is set to $0.001$, while other optimizer hyperparameters remain at their default values. 
The optimization process is typically conducted for a maximum of $T_{\max} = 2000$ iterations.

We adopt distinct training strategies depending on the experimental goal. 
For phase discovery, we employ strategy I to facilitate escape the disordered phase. Specifically, we set $C=10$ and define the penalty weight $w_{\text{p}}(t)$ as a linear decay schedule
\begin{equation*}
w_{\text{p}}(t)=\begin{cases}
w_{\max}, & t < 250, \\
w_{\max} \cdot \frac{500-t}{250}, & 250 \leqslant t < 500, \\
0, & t \geqslant 500,
\end{cases}
\end{equation*}
where $w_{\max}=0.1$. 
When the objective is to locate a specific configuration, we utilize strategy II, and the guidance weight $w_{\text{d}}(t)$ is maintained at a constant value of $0.1$ for the initial 500 steps to anchor the solution, after which it is set to 0 for the remaining iterations to enable purely physical relaxation.

\paragraph{Phase identification and statistics} To statistically analyze the model's capability to discover various metastable phases, we adopt an energy-based identification protocol.
Prior to training, we generate a library of reference ground truth solutions by considering a predefined set of candidate phases $S$ tailored to the spatial dimension.
For 2D systems, the candidate set is $S=\{\text{LAM}, \text{HEX}\}$, whereas for 3D systems, it encompasses a broader spectrum of complex structures, $S=\{\text{LAM}, \text{HEX}, \text{DG}, \text{BCC}, \text{FCC}, \text{A15}, \sigma\}$. 
For each phase $s \in S$, we construct the initial guess $(\phi_{\text{s}}^0, \omega_{\text{s}}^0)$ based on reciprocal vectors \cite{mcclenagan2019landau} and compute the converged stationary state $(\phi_{\text{s}}, \omega_{\text{s}})$ using the high-precision numerical solvers \eqref{eq:solver} to obtain the reference free-energy density $f_{\text{ref}}^{(s)} = f(\phi_{\text{s}}, \omega_{\text{s}})$.

For a predicted solution $(\phi_{\theta}, \bm{\beta}^{-1})$ obtained from the neural network training, we identify its phase by comparing its energy density with the reference library. 
We define the relative energy difference $\mathcal{E}_{\text{diff}}^{(s)}$ with respect to phase $s$ as
\begin{equation*}
\mathcal{E}_{\text{diff}}^{(s)} = \frac{\left| f(\phi_{\theta}, \bm{\beta}^{-1}) - f_{\text{ref}}^{(s)} \right|}{\left| f_{\text{ref}}^{(s)} \right|}.
\end{equation*}
If there exists a reference phase subject to $\mathcal{E}_{\text{diff}}^{(s)} < 10^{-5}$, we classify the prediction as the corresponding phase $s$. 
To ensure statistical robustness, for general phase discovery tasks, we conduct 100 independent training runs with random initializations.

\paragraph{Evaluation metrics} To quantitatively assess the accuracy and physical validity of the predicted solutions, we further introduce three distinct evaluation metrics. 
First, since any local minimum must satisfy the stationary condition, we evaluate the residual of the E-L equation,
\begin{equation*}
\mathcal{E}_{\text{res}} = \left\|\mathcal{P}\frac{\delta f}{\delta\phi}\right\|^2 = \frac{1}{|\Omega|}\int_{\Omega}\left(\mathcal{P}\frac{\delta f}{\delta\phi}\right)^2\,\mathrm{d}\bm{x},
\end{equation*}
which measures the physical consistency of the network output.
Next, to assess the convergence accuracy relative to the ground truth, we use the traditional numerical solver as a refinement tool. 
We take the network's prediction $(\phi_{\theta}, \bm{\beta}^{-1})$ as an initial guess and run the solver \eqref{eq:solver} to convergence, obtaining a refined reference solution $(\phi^*, \bm{\omega}^*)$. 
The state approximation error $\mathcal{E}_{\phi}$ is defined as the $L^2$ distance between the prediction and this refined state
\begin{equation*}
\mathcal{E}_{\phi} = \|\phi_{\theta}-\phi^*\|^2 = \frac{1}{|\Omega|}\int_{\Omega}(\phi_{\theta}-\phi^*)^2\,\mathrm{d}\bm{x}.
\end{equation*}
Finally, we measure the accuracy of the optimized domain geometry compared to the refined domain parameters $\bm{\omega}^*$ using the relative domain error $\mathcal{E}_{\text{dom}}$,
\begin{equation*}
\mathcal{E}_{\text{dom}} = \frac{\|\bm{\beta}^{-1}-\bm{\omega}^*\|_2^2}{\|\bm{\beta}^{-1}\|_2^2}.
\end{equation*}

\subsection{Escaping the disordered phase}

\begin{figure}[!t]
\centering
\makesubfig{fig:baseline_a}
\makesubfig{fig:baseline_b}
\makesubfig{fig:baseline_c}
\makesubfig{fig:baseline_d}
\includegraphics[width=\linewidth]{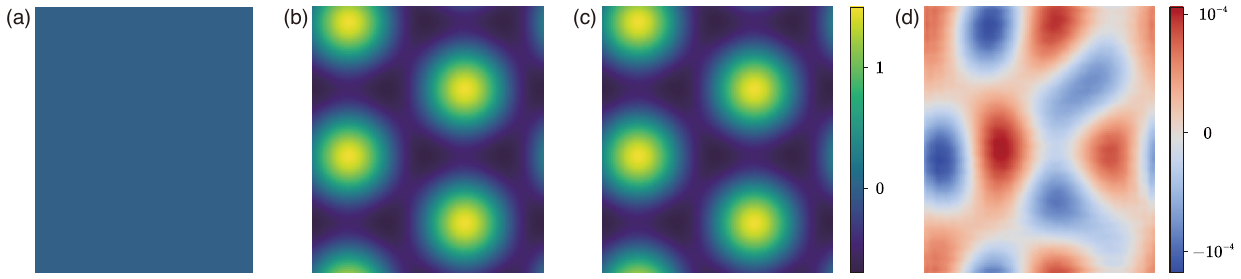}
\caption{Comparison of phase discovery results on the 2D LB model. 
(a) The standard DRM gets trapped in the disordered phase. 
(b) Our method successfully nucleates the hexagonal phase $\phi_{\text{pred}}$. 
(c) The high-precision reference solution $\phi_{\text{ref}}$. 
Panels (a)--(c) share the common colorbar located to the right of panel (c). 
(d) The pointwise difference map $\phi_{\text{pred}} - \phi_{\text{ref}}$, showing the spatial distribution of the signed error with a magnitude of $\mathcal{O}(10^{-4})$, corresponding to the colorbar on the far right.}
\label{fig:baseline}
\end{figure}

We validate our proposed GeoDVF on the 2D LB model with parameters $(\tau, \gamma) = (0.001, 0.6)$, where the disordered phase $\phi=0$ is a local minimum, and the HEX phase corresponds to the global minimum. 
This setup serves as a benchmark to evaluate the method's capability to destabilize and escape the basin of attraction of the disordered phase. 
We compare our proposed method against the standard DRM and traditional numerical solvers. The results are summarized in Figure~\ref{fig:baseline}.

First, we employ the standard DRM by setting the penalty weight $w_{\text{p}}=0$. 
As expected, the optimization from small random initialization rapidly stagnates at the disordered phase $\phi=0$, failing to nucleate an ordered phase (see Figure~\figsubref{fig:baseline}{fig:baseline_a}). 
As discussed in Section~\ref{sec:lbmodel}, traditional iterative solvers are highly sensitive to the initial domain size. 
When the computational domain is incommensurate with the proper periodicity, these methods typically fail to converge to the global HEX minimum but get trapped in metastable states (see Figure~\ref{fig:LB2d} and Table~\ref{tab:sensitivity}).
By comparison, our GeoDVF equipped with the warmup penalty successfully escapes from the disordered phase. 
Driven by the joint optimization mechanism, the system consistently converges to the HEX phase and achieves the theoretical minimum energy, demonstrating robustness against random initializations and domain constraints (see Figure~\figsubref{fig:baseline}{fig:baseline_b}).
We emphasize that the primary objective of GeoDVF is spontaneous phase discovery rather than accelerating a single static simulation. 
Although training a neural network from scratch may require more wall-clock time than running a traditional solver from a given initial guess on a fixed grid, GeoDVF completely eliminates the need to construct initial guesses and domain sizes. 
This paradigm shift significantly enhances the overall exploration efficiency for discovering complex ordered phases without a priori knowledge.

Furthermore, since machine learning approaches and traditional numerical algorithms typically operate at different precision levels, GeoDVF is best utilized as a powerful structural predictor.
Quantitatively, the final state of our model reaches a low E-L residual of $\mathcal{E}_{\text{res}}=7.1\times10^{-3}$, indicating that the network has effectively converged to a stationary point of the LB energy. 
To precisely measure the pattern, the output of GeoDVF can be effortlessly fed as an initial condition into the high-precision numerical solver~\eqref{eq:solver}.
The solution then converges to a high-precision reference solution $\phi^*$ within a few iterations (see Figure~\figsubref{fig:baseline}{fig:baseline_c}). 
To visually inspect the prediction quality, we plot the pointwise difference map in Figure~\figsubref{fig:baseline}{fig:baseline_d}, and the errors are uniformly distributed with a magnitude of $\mathcal{O}(10^{-4})$, indicating excellent local agreement. 
This accuracy is further confirmed by the global relative errors $\mathcal{E}_{\phi}=2.1\times10^{-9}$ and $\mathcal{E}_{\text{dom}}=2.6\times10^{-10}$.

\subsection{Ablation study}
To investigate the efficiency of the scheduler function $w_{\text{p}}(t)$ and its peak magnitude $w_{\max}$, we conduct an ablation study on the 2D LB model with parameters $(\tau, \gamma) = (0.001, 0.6)$. 
We compare two distinct scheduler functions:
\begin{enumerate}
\item \textbf{Hard cutoff:} The penalty weight is held constant at $w_{\text{p}}(t) = w_{\max}$ for the first 500 steps, and then abruptly dropped to zero. 
This represents a sudden removal of the external forcing (see Figure~\figsubref{fig:ablation_scheduler}{fig:wt_a}).
\item \textbf{Linear decay:} The penalty weight remains constant at $w_{\text{p}}(t) = w_{\max}$ for the initial 250 steps, then linearly decays to zero between steps 250 and 500. 
This provides a smoother transition for the optimization landscape (see Figure~\figsubref{fig:ablation_scheduler}{fig:wt_b}).
\end{enumerate}

\begin{figure}[!t]
\centering
\makesubfig{fig:wt_a}
\makesubfig{fig:wt_b}
\includegraphics[width=\linewidth]{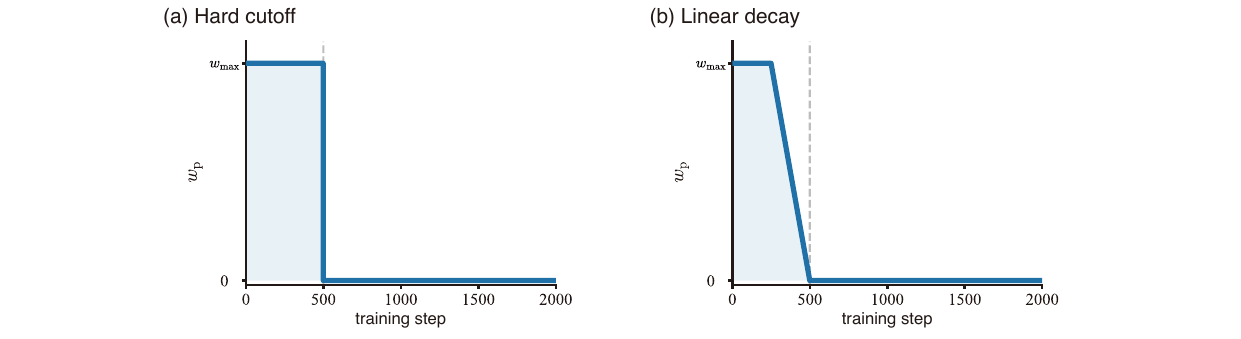}
\caption{Illustration of two scheduler functions. 
(a) Hard cutoff. (b) Linear decay. 
The blue regions indicate when the spectral forcing is active.}
\label{fig:ablation_scheduler}
\end{figure}

For each scheduler type and values of $w_{\max}$, we conduct 100 independent training runs with random initializations. 
Table~\ref{tab:scheduler} reports the success rate, defined as the percentage of runs that correctly converge to the global HEX phase. 
In addition to the standard solution errors ($\mathcal{E}_{\text{res}}, \mathcal{E}_{\phi}, \mathcal{E}_{\text{dom}}$), we specifically report the relative energy difference $\mathcal{E}_{\text{diff}}$ for the successful cases (i.e., the HEX phase) to verify that the converged solutions indeed achieve the theoretical minimum free-energy density.
The baseline case ($w_{\max}=0$) consistently fails to escape the disordered phase (100\% failure rate), while introducing the penalty term significantly improves the success rate for both schedules. 
However, a distinct performance gap appears as $w_{\max}$ increases. 
While the hard-cutoff scheduler performs comparably to the linear-decay scheduler at low magnitudes ($w_{\max} \leqslant 0.01$), it suffers a catastrophic performance drop at higher magnitudes. 
For instance, at $w_{\max}=0.1$, the success rate of hard cutoff drops to 18\%, whereas linear decay achieves its peak performance of 71\%.
We attribute this to the optimization shock caused by the abrupt removal of the penalty in the hard cutoff. 
When a large forcing term is immediately removed, the sudden shift in the loss landscape can eject the trajectory from the attraction basin of the HEX phase back to the disordered phase. 
In contrast, the linear-decay scheduler gradually relaxes the system, allowing it to settle stably into the global minimum even with strong initial forcing.
Based on these findings, we observe that the linear decay offers superior robustness, particularly at larger $w_{\max}$. 
Consequently, we adopt the linear-decay scheduler with $w_{\max}=0.1$ as the default setting for all subsequent experiments, as it yields a high success rate and good accuracy.

\begin{table}[!t]
\centering
\begin{threeparttable}
\caption{Comparison of phase discovery success rates and error metrics under different penalty schedules and peak magnitudes.}
\label{tab:scheduler}
\begin{tabular}{ccccccccc}
\toprule
\multirow{2}{*}{Schedule} & \multirow{2}{*}{$w_{\max}$} & \multicolumn{3}{c}{Convergence rate} & \multicolumn{4}{c}{Error metrics} \\
\cmidrule(lr){3-5} \cmidrule(lr){6-9} 
 &  & HEX & HEX\tnote{*} & disorder & $\mathcal{E}_{\text{diff}}$ & $\mathcal{E}_{\text{res}}$ & $\mathcal{E}_{\phi}$ & $\mathcal{E}_{\text{dom}}$ \\
\midrule

Baseline & $0$ & 0\% & 0\% & 100\% & - & - & - & - \\
\midrule

\multirow{5}{*}{Hard cutoff} 
 & 0.0001 & 0\%  & 3\%  & 97\% & - & - & - & - \\
 & 0.001  & 56\% & 44\% & 0\%  & $3.7\times10^{-6}$ & $5.1\times10^{-4}$ & $4.1\times10^{-8}$ & $2.4\times10^{-11}$ \\ 
 & 0.01   & 70\% & 30\% & 0\%  & $1.8\times10^{-6}$ & $1.3\times10^{-3}$ & $6.6\times10^{-10}$ & $3.7\times10^{-11}$ \\
 & 0.1    & 18\% & 0\%  & 82\% & $9.5\times10^{-6}$ & $7.4\times10^{-2}$ & $6.1\times10^{-8}$ & $2.8\times10^{-10}$ \\
 & 0.5    & 0\%  & 0\%  & 100\%& - & - & - & - \\
\midrule

\multirow{5}{*}{Linear decay} 
 & 0.0001 & 0\%  & 0\%  & 100\% & - & - & - & - \\
 & 0.001  & 56\% & 44\% & 0\%   & $1.2\times10^{-5}$ & $5.7\times10^{-4}$ & $3.5\times10^{-7}$ & $2.4\times10^{-10}$ \\
 & 0.01   & 70\% & 30\% & 0\%   & $1.5\times10^{-6}$ & $1.1\times10^{-3}$ & $1.7\times10^{-10}$ & $2.2\times10^{-11}$ \\
 & 0.1    & 71\% & 28\% & 1\%   & $4.9\times10^{-6}$ & $1.2\times10^{-2}$ & $8.6\times10^{-9}$ & $2.0\times10^{-10}$ \\
 & 0.5    & 62\% & 22\% & 16\%  & $6.6\times10^{-6}$ & $1.5\times10^{-2}$ & $6.8\times10^{-9}$ & $2.0\times10^{-10}$ \\
\bottomrule
\end{tabular}

\begin{tablenotes}
\small
\item[] The symbol `-' indicates that the metrics are not applicable for these cases.
\end{tablenotes}
\end{threeparttable}
\end{table}

Furthermore, we investigate the impact of the warmup strategy when the disordered phase is unstable. 
We conduct numerical experiments in the regime $(\tau, \gamma) = (-0.4, 0.4)$, where the HEX phase is stable and the LAM phase is metastable. 
In the absence of the warmup penalty (i.e., standard DRM, $w_{\max}=0$), the method consistently fails to locate the HEX phase within 100 independent runs using small random initializations. 
Although the metastable LAM phase can be identified, the computational domain often becomes elongated, reflecting the translational invariance of the LAM phase along one direction.
In contrast, employing the warmup penalty with a linear-decay scheduler of $w_{\max}=0.1$, our proposed method successfully identifies both the HEX phase and the LAM phase with a geometrically regular domain size, demonstrating enhanced exploration capability. 
These results indicate that the warmup penalty enhances robustness by facilitating the escape from the disordered phase, enlarging the range of attainable metastable states, and promoting geometric regularization through the enforcement of a characteristic length scale during training.

\subsection{Global minimum selection}
We further investigate the behavior of GeoDVF in a competitive regime where the disordered phase is unstable, and the 2D LB model competes between two ordered phases: the LAM and HEX phases.
Specifically, we vary the parameters $(\tau, \gamma)$ across the phase boundary such that the global minimum transitions from LAM to HEX.
Notably, we find that our method exhibits a strong preference for converging to the global minimum, effectively avoiding a trap in the local minimum.

To explain the physical mechanism behind this selectivity, we employ the LBSD method \cite{zhang2025exploring} to compute the transition pathway and the transition state connecting the LAM and HEX phases.
These calculations are conducted on a fixed computational domain $\Omega = [0, 8\pi] \times [0, 16\pi/\sqrt{3}]$, which is chosen to properly accommodate the periodicities of both phases.
The transition pathways and corresponding energy landscapes are visualized in Figure~\ref{fig:transition_pathway}. A crucial physical insight is observed as the system moves away from the phase boundary—either by decreasing $\gamma$ (favoring LAM, see Figure~\figsubref{fig:transition_pathway}{fig:trans_a}) or increasing $\gamma$ (favoring HEX, see Figure~\figsubref{fig:transition_pathway}{fig:trans_d}). 
In both scenarios, the transition state structurally converges to the local minimum, differing only by a localized region highlighted by the white circles in Figure~\ref{fig:transition_pathway}. 
This difference characterizes the critical nucleus required to initiate the phase transition. 
When the transition state is gradually close to the local minimum, it implies that the critical nucleus for the global minimum is small and localized, so the energetic barriers for nucleation are significantly reduced.
Our method with the warmup penalty, driven by the spectral forcing, induces sufficient fluctuations to form this small critical nucleus, thereby helping the system to cross the saddle point and settle into the global minimum basin.

This mechanism is quantitatively supported by the statistical results in Table~\ref{tab:gm_selection}, as a correlation between the nucleation energy barrier $\Delta f$ and the phase discovery rate is observed. 
As $\gamma$ increases from 0.3 to 0.45, the barrier for forming the HEX phase decreases from $1.6 \times 10^{-2}$ to $4.9 \times 10^{-4}$. 
Correspondingly, the success rate of discovering the HEX phase surges from 1\% to nearly 72\% (including HEX*). 
In contrast, at $\gamma=0.3$, the LAM phase is the global minimum with a low energy barrier, and our method successfully identifies it in 97\% of the runs. 
These findings suggest that GeoDVF is particularly effective in locating the global minimum when the critical nucleation barrier is relatively small, effectively leveraging the easy nucleation pathways inherent in the energy landscape.

\begin{figure}[!t]
\centering
\makesubfig{fig:trans_a}
\makesubfig{fig:trans_b}
\makesubfig{fig:trans_c}
\makesubfig{fig:trans_d}
\includegraphics[width=\linewidth]{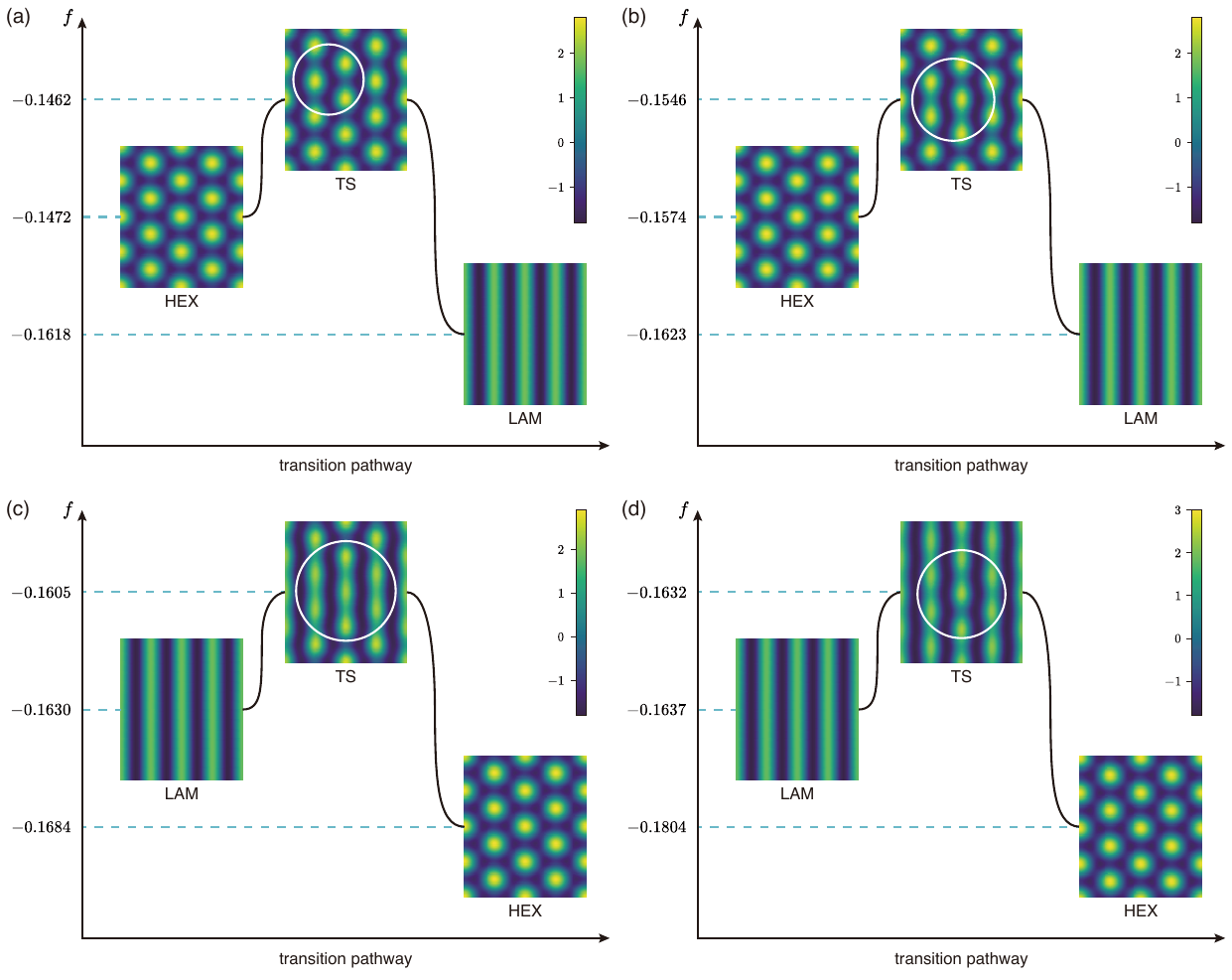}
\caption{Transition pathways between the LAM and HEX phases for (a) $(\tau, \gamma) = (-0.4, 0.3)$, (b) $(-0.4, 0.35)$, (c) $(-0.4, 0.4)$, and (d) $(-0.4, 0.45)$. 
In each plot, the vertical axis denotes the free-energy density $f$, and the pathway connects the local minimum (left), the transition state (TS, middle), and the global minimum (right).
Note that in (a) and (b), the LAM phase is the global minimum, whereas in (c) and (d), the HEX phase becomes the global minimum.
The white circles highlight the critical nucleus region within the transition states.
}
\label{fig:transition_pathway}
\end{figure}

\begin{table}[!t]
\centering
\begin{threeparttable}
\caption{Quantitative analysis of phase selection. A lower energy barrier for nucleation associates with a higher success rate of locating the global minimum.}
\label{tab:gm_selection}

\begin{tabular}{ccccc}
\toprule
\multirow{2}{*}{Property} & \multicolumn{4}{c}{$(\tau, \gamma)$} \\ 
 \cmidrule(lr){2-5}
 & $(-0.4, 0.3)$ & $(-0.4, 0.35)$ & $(-0.4, 0.4)$ & $(-0.4, 0.45)$ \\
\midrule

\multicolumn{5}{l}{\textit{Energy landscape}} \\ 
Global minimum & LAM & LAM & HEX & HEX \\
Energy barrier $\Delta f$ (LAM $\to$ HEX) & $1.6\times10^{-2}$ & $7.7\times10^{-3}$ & $2.4\times10^{-3}$ & $4.9\times10^{-4}$ \\
HEX critical nucleus & large & \multicolumn{2}{c}{$\xrightarrow{\makebox[3cm]{}}$} & small \\
Energy barrier $\Delta f$ (HEX $\to$ LAM) & $1.0 \times 10^{-3}$ & $2.8 \times 10^{-3}$ & $7.9 \times 10^{-3}$ & $1.7 \times 10^{-2}$ \\
LAM critical nucleus & small & \multicolumn{2}{c}{$\xrightarrow{\makebox[3cm]{}}$} & large \\

\midrule

\multicolumn{5}{l}{\textit{Numerical convergence (100 runs)}} \\ 
HEX phase\tnote{a} & $1\%$ & $18\%$ & $38\%\ (+7\%)$ & $58\%\ (+14\%)$ \\
LAM phase & $97\%$ & $77\%$ & $45\%$ & $20\%$ \\
Failure / others\tnote{b} & $2\%$ & $5\%$ & $10\%$ & $8\%$ \\
\bottomrule
\end{tabular}

\begin{tablenotes}
\small
\item[a] Values in brackets indicate the percentage of distorted hexagonal phases (HEX*). 
\item[b] Fail to converge or converge to other high-energy metastable states.
\end{tablenotes}
\end{threeparttable}
\end{table}

\subsection{Discovery of complex 3D ordered structures}
We extend our investigation to the 3D LB model to demonstrate the representation capability of our proposed method on a complicated energy landscape. 
Various 3D ordered structures are visualized in Figure~\ref{fig:LB3d}, where the red regions represent where the order parameter $\phi$ is high in the identified phases.

Using the warmup penalty starting from random initializations, GeoDVF successfully identifies the stable phases LAM, HEX, BCC, and FCC phases. 
Remarkably, beyond these phases, our method is able to spontaneously nucleate the A15 phase without prior knowledge. 
The discovery of this spherical packing structure, which consists of 8 spheres per unit cell, highlights the efficiency of the warmup penalty in guiding the system out of simple metastable traps to organize complex multi-particle packing. 
Quantitatively, as reported in Table~\ref{tab:LB3d}, the relative energy differences $\mathcal{E}_{\text{diff}}$ for these phases are on the order of $10^{-7}$ to $10^{-6}$, confirming that the nucleated structures have converged to the theoretical ground states with high fidelity.

However, certain exotic phases, such as DG and the Frank--Kasper $\sigma$ phase, possess narrow basins of attraction far away from the disordered phase. 
It is important to note that the difficulty in locating these phases is inherent to the energy landscape rather than a limitation of the optimizer; indeed, traditional numerical solvers also fail to nucleate these complex topologies from random noises and invariably require specific initializations. 
To capture these topologically intricate structures, we employ guided initialization by utilizing the approximations provided in \cite{mcclenagan2019landau} as initial guesses. 
These phases present distinct challenges: the DG phase is a bicontinuous network characterized by two interpenetrating but non-intersecting labyrinths, while the $\sigma$ phase is a complex packing featuring a unit cell with up to 30 spheres.

Our results demonstrate that our network successfully locates high-precision stationary points from low-quality initializations. 
Table~\ref{tab:LB3d} shows that even for these complex inputs, the method achieves $\mathcal{E}_{\text{diff}} \sim 10^{-6}$ and $\mathcal{E}_{\phi} \sim 10^{-10}$, indicating that the final solutions are not merely reproductions of the initial guesses but have been optimized to the precise energetic minima. 
The ability of the neural network to accurately resolve the continuous curvature of the DG interfaces and the precise packing arrangement of the $\sigma$ phase confirms that our architecture possesses sufficient representational capacity to approximate high-frequency features and complex topologies in 3D configuration spaces.

\begin{figure}[!t]
\centering
\includegraphics[width=\linewidth]{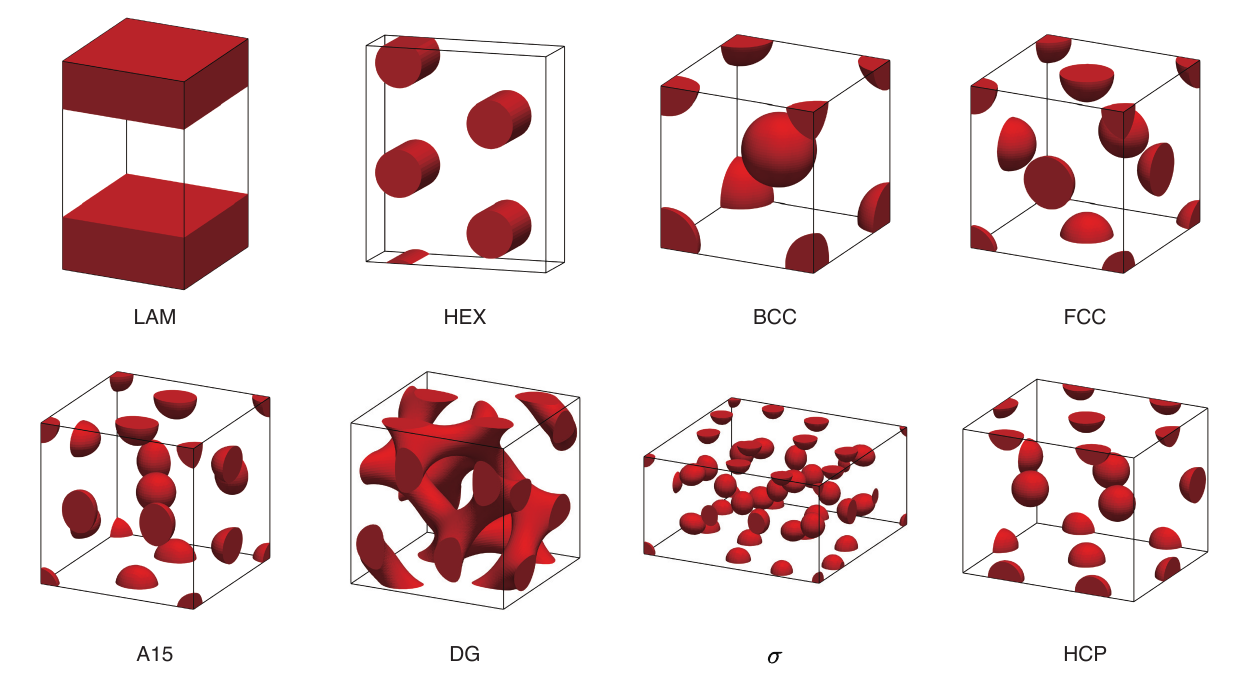}
\caption{3D ordered structures identified by our GeoDVF. 
The gallery displays a diverse set of local minima in the 3D LB model: 
(top row) LAM, HEX, BCC, FCC; 
(bottom row) A15, DG, $\sigma$, and HCP. 
The red volumetric regions correspond to where the order parameter is high.
The LAM, HEX, BCC, FCC, A15 and HCP phases can be successfully nucleated from random initializations using strategy I. 
}
\label{fig:LB3d}
\end{figure}

\begin{table}[!t]
\centering
\begin{threeparttable}
\caption{Quantitative evaluation of the discovered 3D phases. }
\label{tab:LB3d}

\begin{tabular}{cccccc}
\toprule
Phase & $(\tau, \gamma)$ & $\mathcal{E}_{\text{diff}}$ & $\mathcal{E}_{\text{res}}$ & $\mathcal{E}_{\phi}$ & $\mathcal{E}_{\text{dom}}$ \\ 
\midrule

% Strategy I Group
\multicolumn{6}{l}{\textit{Strategy I: warmup penalty}} \\
LAM & $(-0.4, 0.22)$ & $1.8\times 10^{-7}$ & $2.1\times 10^{-1}$ & $6.9\times 10^{-11}$ & $5.7\times 10^{-12}$ \\ 
HEX & $(-0.3, 0.8)$ & $3.8\times 10^{-7}$ & $3.2\times 10^{-1}$ & $2.4\times 10^{-10}$ & $1.0\times 10^{-11}$ \\ 
BCC & $(-0.008, 0.8)$ & $6.9\times 10^{-6}$ & $2.6\times 10^{-2}$ & $6.7\times 10^{-10}$ & $8.7\times 10^{-10}$ \\ 
FCC & $(0.1, 1.6)$ & $1.3\times 10^{-6}$ & $1.5\times 10^{-2}$ & $7.2\times 10^{-10}$ & $4.6\times 10^{-11}$ \\ 
A15 & $(-0.008, 1.2)$  & $3.2\times 10^{-6}$ & $6.7\times 10^{-3}$ & $7.4\times 10^{-10}$ & $2.1\times 10^{-11}$ \\ 
HCP & $(-0.008, 0.8)$ & $5.7\times 10^{-6}$ & $2.6\times 10^{-3}$ & $5.2\times 10^{-8}$ & $1.6\times 10^{-10}$ \\ 

\midrule

% Strategy II Group
\multicolumn{6}{l}{\textit{Strategy II: guided initialization}} \\
DG & $(-0.4, 0.4)$     & $8.6\times 10^{-7}$ & $3.1\times 10^{-3}$ & $9.8\times 10^{-11}$ & $6.0\times 10^{-13}$ \\
$\sigma$ & $(-0.008, 1.0)$ & $1.5\times 10^{-6}$ & $7.9\times 10^{-4}$ & $2.0\times 10^{-10}$ & $1.5\times 10^{-12}$ \\ 
\bottomrule
\end{tabular}
\end{threeparttable}
\end{table}

While the primary focus of the LB model is the identification of global minima, GeoDVF also exhibits a strong capability to explore the rich energy landscape of the LB model and identify distinct metastable states. 
Capturing these local minima is crucial for understanding the energy landscape of the LB model.
By employing strategy I with random initializations, our method successfully stabilizes the hexagonal close-packed (HCP) phase, which shares identical coordination numbers and very similar energies compared to the FCC phase.

To assess the numerical accuracy of this result, we evaluate the solution against high-precision reference states which are obtained by initializing the numerical solver \eqref{eq:solver} with the network's predicted solutions $\phi_{\text{pred}}$ (see Table~\ref{tab:LB3d}). 
The obtained HCP phase achieves a relative energy difference of $\mathcal{E}_{\text{diff}} = 4.9 \times 10^{-6}$, a state approximation error of $\mathcal{E}_{\phi} = 1.1 \times 10^{-7}$, and a relative domain error of $\mathcal{E}_{\text{dom}}=2.1\times10^{-10}$, indicating exceptional agreement with the theoretical ground state. 
Furthermore, the residual error converges to $\mathcal{E}_{\text{res}} = 3.3 \times 10^{-3}$, confirming that the discovered state is a valid, stationary local minimum on the energy landscape.

To quantitatively analyze the competition between these simple crystalline phases (BCC, FCC, and HCP), we conducted experiments with fixed $\tau = -0.008$ and varying $\gamma$, and the refined energy densities $f$ are reported in Table~\ref{tab:metastable}. 
Although in certain regimes these phases are metastable relative to complex structures (e.g., the $\sigma$ phase at $\gamma=1.0$ and the A15 phase at $\gamma=1.2$, see Table~\ref{tab:LB3d}), our method accurately resolves the energetic hierarchy among the BCC, FCC, and HCP phases. 
As shown in Table~\ref{tab:metastable}, we observe a clear transition in relative stability: 
At $\gamma=0.8$, the BCC phase has the lowest energy among the three. 
At $\gamma=1.0$, the HCP phase becomes more stable than BCC and FCC. 
At $\gamma=1.2$, the FCC phase takes precedence. 
This observed sequence of stability is in excellent agreement with the phase diagrams of the 3D phase-field crystal model reported in \cite{jaatinen2010extended,toth2010polymorphism}. 
This consistency confirms that our method not only locates the global minimum but also correctly captures the subtle energetic competition characteristic of the physical model.

\begin{table}[!t]
\centering
\begin{threeparttable}
\caption{Comparison of free-energy densities $f$ for the competing BCC, FCC, and HCP phases.}
\label{tab:metastable}

\begin{tabular}{ccccc}
\toprule
\multirow{2}{*}{$(\tau, \gamma)$} & \multicolumn{3}{c}{Energy density $f$} & \multirow{2}{*}{Global minimum} \\
\cmidrule(lr){2-4}
 & BCC & FCC & HCP &  \\
\midrule
$(-0.008, 0.8)$ & $\mathbf{-0.02081}$ & $-0.01948$ & $-0.02001$ & BCC \\
$(-0.008, 1.0)$ & $-0.05159$ & $-0.05149$ & $\mathbf{-0.05169}$ & $\sigma$ \\
$(-0.008, 1.2)$ & $-0.1104$ & $\mathbf{-0.1126}$ & $-0.1125$ & A15 \\
\bottomrule
\end{tabular}

\begin{tablenotes}
\small
\item Bold values indicate the lowest energy among these three phases.
\end{tablenotes}
\end{threeparttable}
\end{table}

\section{Conclusion and discussion}
\label{sec:final}
In this work, we proposed GeoDVF for discovering ordered structures in the LB model. 
By parameterizing the order parameter through a neural network and treating the domain geometry as a set of learnable variables, our method enables the simultaneous variational optimization of the physical field and the computational domain. 
This joint optimization approach effectively eliminates the artificial stress caused by domain mismatch, removing the heavy dependence on precise geometric initial guesses that limits traditional alternating iterative solvers.

Another key contribution of this study is the warmup penalty strategy. 
Our analysis and experiments show that this strategy effectively destabilizes the disordered phase and promotes geometric regularization to prevent uncontrolled domain stretching during the early stage of training, which allows the automatic formation of complex 3D phases starting purely from random initializations.
Beyond identifying global minima, we successfully found distinct metastable states, and the ability to resolve these structures highlights the method's strength in exploring diverse solutions within a complex energy landscape.

However, limitations remain regarding the exploration of extremely narrow attraction basins.
For topologically intricate phases such as the DG and $\sigma$ phases, spontaneous nucleation from random initializations remains challenging. 
Although we demonstrated that a guided initialization strategy can successfully refine rough approximations of these phases into higher precision stationary states, achieving their spontaneous discovery without prior knowledge requires further research. 
Furthermore, the reliance on discrete Fourier transforms restricts the current framework to orthogonal domains with uniform grids. 
where the system is trapped by the orthogonality of the computational box, preventing full relaxation to non-orthogonal lattice symmetries.
This geometric constraint can also introduce additional local minima, such as the HEX* phase. 

The joint optimization of the infinite-dimensional order parameter and the finite-dimensional geometric parameters occurs in a wide class of pattern-forming systems with intrinsic length scales. 
Beyond the LB model, GeoDVF is applicable to other continuum theories where stress relaxation is critical, such as the self-consistent field theory for block copolymers \cite{matsen1994stable}, the phase-field crystal models \cite{elder2004modeling}, and Landau-de Gennes theory for liquid crystals \cite{de1993physics}. 
Extending our GeoDVF to these complex systems represents a promising direction for future investigation.

\section*{CRediT authorship contribution statement}
\textbf{Yuchen Xie:} Conceptualization, Formal analysis, Investigation, Methodology, Software, Validation, Visualization, Writing - original draft.
\textbf{Jianyuan Yin:} Conceptualization, Formal analysis, Methodology, Writing - review \& editing.
\textbf{Lei Zhang:} Conceptualization, Funding acquisition, Resources, Supervision, Writing - review \& editing.

\section*{Declaration of competing interest}
The authors declare that they have no known competing financial interests or personal relationships that could have appeared to influence the work reported in this paper.

\section*{Data availability}
The codes are available from the authors upon reasonable request.

\bibliographystyle{elsarticle-num} 
\bibliography{references}

\end{document}